\documentclass[11pt]{article}
\usepackage{amsfonts}
\usepackage{amsmath, amsfonts, amsthm, amssymb, amscd, mathrsfs}
\usepackage{textcomp,xspace}
\usepackage[mathcal]{euscript}

\usepackage{color}
\usepackage {a4wide}
\theoremstyle{plain}

\newtheorem*{pmt}{Spacetime Positive Energy Theorem}

\newtheorem{theorem}{Theorem}[section]
\newtheorem{proposition}[theorem]{Proposition}

\numberwithin{equation}{section}
\let\oldmarginpar\marginpar
\renewcommand\marginpar[1]{\-\oldmarginpar[\raggedleft\footnotesize #1]%
{\raggedright\footnotesize #1}}

\usepackage{a4wide}

\newcommand \bei {\begin{itemize}}
\newcommand \eei {\end{itemize}}

\newcommand \be         {\begin{equation}}

\newcommand \Vcal {\mathcal V}
\newcommand \Lcal {\mathcal L}

\newcommand \Div {\text{div}}
\newcommand \madm {m_\text{ADM}}

\newcommand \eps \epsilon

\newcommand \coeff \kappa

\newcommand \Ical {\mathcal I}

\newcommand \Mcal {\mathcal M}

\newcommand \tr {\text{tr}}

\newcommand \p  {\prime}

\newcommand \del \partial

\newcommand \RR         {\mathbb R}

\newcommand \LL         {\mathbb L}
\newcommand \ee         {\end{equation}}
\newcommand \la \langle
\newcommand \ra \rangle

\definecolor{myred}{rgb}{0.858, 0.0, 0.0}

\newcommand\blfootnote[1]{%
	\begingroup
	\renewcommand\thefootnote{}\footnote{#1}%
	\addtocounter{footnote}{-1}%
	\endgroup
}
\usepackage{tikz}
\usetikzlibrary{positioning, arrows}
\tikzstyle{variation} = [rectangle, minimum width=3cm, minimum height=1cm, text centered, text width=2.7cm, draw=black,  fill=black!10]
\tikzstyle{process} = [rectangle, minimum width=3cm, minimum height=1cm, text centered, text width=2.7cm, draw=black]
\tikzstyle{arrow} = [->, thick]

\begin{document}
\title{A note on the electrostatic Born--Infeld equation with radial charge density}
\author{The-Cang Nguyen\footnote{E-mail: {\sl alpthecang@gmail.com.}
}}
\date{\today}
\maketitle
\begin{abstract}
In this note, we present a new proof of the solvability of the electrostatic Born–Infeld equation with radial charge, based on the conformal method and the Spacetime Positive Energy Theorem. An advantage of this approach is that the resulting solutions are automatically classical and spacelike.
\end{abstract}
\section{Introduction}
The Born–Infeld Lagrangian is a foundation of field theory, introduced to deal with the failure of classical electrostatic Maxwell theory to satisfy the Principle of Finite Energy. In the vacuum case, it leads to the relation
\begin{subequations} \label{BI}
\begin{align}
&- \Div \bigg( { \nabla u \over \sqrt{1 - |\nabla u|^2}} \bigg) = \rho \qquad \text{in $\RR^n$}, \label{MCequation}
\\
&\lim_{|x| \to +\infty} u = 0, \label{uto0}
\end{align}
\end{subequations}
where $\rho$ is the charge density and $u$ represents the electric potential. This equation is called the \textit{electrostatic Born–Infeld equation} and has received significant attention in recent years. The most common approach to solving this problem is the variational method, in which the equation is interpreted as the Euler–Lagrange equation of the functional
\begin{equation}\label{def I}
I(\phi) := \int_{\RR^n} \bigg(1 - \sqrt{1 - |\nabla \phi|^2}  \bigg) dx - \langle \rho, \phi \rangle,
\end{equation}
and one looks for a weak solution $u$ to \eqref{BI} as the minimizer of $I$. Until now, there have been two main existence results obtained for \eqref{BI} under assumptions of $\rho$: the first one given in \cite{BonheureAveniaPomponio BIEquationRadialCharge} requires that $\rho$ is radially distributed, while the second one \cite{BonheureIacopetti BIEquationSmallCharge} relies on a smallness condition on $\rho$.
\\
\\
In this work, we are interested in the case of radial charge densities and provide a new proof of existence in view of general relativity. The key observation is that, in classical relativity, the equation \eqref{MCequation} is exactly the mean curvature equation in Lorentz–Minkowski spacetime $\LL^{n+1}$. Indeed, any spacelike hypersurface of  $\LL^{n+1}$ can be written as a graph $\Sigma = \{(x,u(x)) ~|~ x \in \RR^n\}$,  and if $\rho$ denotes the mean curvature of its second fundamental form, then the pair $(\rho, u)$ automatically satisfies \eqref{MCequation}. In addition, since the induced metric $h$ of $\Sigma$ in this case is expressed by
$$
h_{ij} = \delta_{ij} - \del_i u \del_j u, 
$$
if $(\Sigma, h)$ is sufficiently asymptotically flat so that $h - \delta_{\text{Euc}} = O(r^{-p})$ at infinity for some $p > 2$, it follows that the graph function $u$ converges to a constant at infinity, which we can normalize to zero, and so $u$ also holds \eqref{uto0}.
\\
\\
The geometric insight proposes a different approach that instead of solving \eqref{BI} directly, we will aim to construct an asymptotically flat (AF) spacelike hypersurface of $\LL^{n+1}$ with prescribed mean curvature $\rho$.
\\
{\small \begin{figure}[h]
	\centering
\begin{tikzpicture}[node distance=3cm]
	
	\node (input) [variation] {Functional $I$};
	\node (output) [variation, below of=input] {Born-Infeld equation};
	\node (process1a) [process, right=of output] {Spacelike hypersurface $(\Sigma, h, k)$ of Minkowski};
	\node (process2a) [process, right=of process1a] {Constraint solution $(M,g,k)$ with $\tr_g k = \rho$};
	\node (start) [process, above of=process2a] {Free data $\rho$};

	\draw [arrow] (input) --  node[anchor=west]{minimizer of $I$}(output);
	\draw [arrow] (process1a) -- node[anchor=south]{mean curvature}node[anchor=north]{$\rho$} (output);
	\draw [arrow] (process2a) -- node[anchor=south]{zero mass}node[anchor=north]{Spacetime PET} (process1a);
	\draw [arrow] (start) --  node[anchor=east]{conformal method}(process2a);
\end{tikzpicture}
\caption{Variation method and an Approach based on general relativity}
\label{figure two methods}
\end{figure}}
\\
To deal with this, we rely on two tools from general relativity: the conformal method and the Spacetime Positive Energy Theorem (PET). More precisely, as illustrated in Figure~\ref{figure two methods}, which presents a schematic of the two methods discussed, it is well-known by the rigidity part of the Spacetime PET that if an AF manifold $(M,g)$, together with a tensor $k$, forms a vacuum initial data set for the Cauchy problem in general relativity and has zero ADM mass, then $(M,g)$ is isometric to a spacelike hypersurface of $\LL^{n+1}$ with second fundamental form $k$. Therefore, to construct the desired hypersurface, we only need to find a vacuum AF initial data set $(M,g,k)$ such that $\tr_g k = \rho$ and its mass is zero. For this task, it is natural to think of using the conformal method, a traditional to construct initial data sets from scratch with prescribed mean curvature. When $\rho$ is radial, as we will see below, this approach yields a feasible and quite simple construction. More precisely, working within the weighted Hölder spaces introduced in Subsection \ref{subsection weighted Holder space}, we obtain the following result.
\begin{theorem} \label{main theorem} Let $\rho$ be an arbitrary radial function in $C^{1, \alpha}_{ - q} (\RR^n)$ with $\alpha \in (0,1)$ and
\begin{equation} \label{q decay at infinity} 
q >  \left\{\begin{array}{lr}
	\max \big\{2, {n \over 2} \big\} & \text{\quad if $3 \le n < 8$,}
	\\
	n - 2 & \text{\quad if $n \ge 8$~~~~~.}
\end{array}\right.
\end{equation}
Then there exists an asymptotically flat spacelike hypersurface $(\Sigma, h)$ of the Lorentz-Minkowski spacetime $\LL^{n+1}$ such that $h - \delta_{\text{Euc}} \in C^{3,\alpha}_{-2q + 2}$ and its mean curvature equals to $\rho$. In particular, the electrostatic Born-Infeld equation \eqref{BI} admits at least one spacelike solution $u \in C^{2,\alpha}_{-q + 2}$.  
\end{theorem}
In comparison with the result in \cite{BonheureAveniaPomponio BIEquationRadialCharge}, our result is slightly weaker, since we require Hölder regularity for $\rho$ in $\RR^n$. Nevertheless, the approach has an advantage that the resulting solutions in Theorem \ref{main theorem} are spacelike and classical, whereas those derived via the variational method are only weak solutions. Moreover, our method avoids the difficulty in the variational approach, namely, showing that a minimizer of $I$ solves \eqref{BI}.
\\
\\
It is also worth noting that, as mentioned in Section \ref{section proof of main theorem}, Theorem \ref{main theorem} can be further strengthened by assuming only that $q>2$. However, we prefer the stated version since it seems to allow a generalization to a larger class of $\rho$, rather than being restricted to the radial case.
\\
\\
The outline of the article is as follows. In Section \ref{section preliminaries}, we give a brief summary of the conformal method and the Spacetime PET. In Section \ref{section proof of main theorem}, the proof of Theorem \ref{main theorem} based on these tools is established.
\paragraph*{Acknowledgments.} The author would like to thank the anonymous referee of a previous version of this article for recommending Birkhoff’s theorem and for providing many helpful comments. Part of the results in this article were obtained during the author’s postdoctoral fellowship funded by the French National Research Agency (ANR) at the Institut Montpelliérain Alexander Grothendieck during 2020–2021.
\section{Preliminaries} \label{section preliminaries}
In this section, we review some standard facts on the conformal method, weighted Hölder spaces and the Spacetime PET. While these notions are well defined on general manifolds, for our purposes we restrict our attention to the Euclidean space $(\RR^n , \delta_{\text{Euc}})$. For a treatment of the general case, we refer the reader to \cite{Bartnik, Diltsthesis}.
\subsection{Conformal method} \label{subsection conformal method}
An AF manifold $(M , g )$ of $n$ dimensions, with $n \ge 3$, coupled with a symmetric $(0,2)-$tensor $k$, is called a vacuum initial data set for the Cauchy problem in general relativity if $(M, g, k)$ satisfies the system
\begin{subequations} \label{constraint}
	\begin{align}
		R_g - | k |_g^2 +(\tr_g k )^2 &= 0  &&\text{{\footnotesize [Hamiltonian constraint}]}
		\label{hamilton}
		\\
		\Div_g \big( k - (\tr_g k ) g \big) & = 0, &&\text{[{\footnotesize Momentum constraint}]}
		\label{momentum}
	\end{align}
\end{subequations}
where $R_g$ is the scalar curvature of $g$. These equations are called the \textit{vacuum Einstein constraint equations} and the study of \eqref{constraint} is a topical issue. Using the conformal method, to construct solutions to \eqref{constraint} grounded in Euclidean space, let $\rho$  be a scalar function and let $\sigma$ be a trace-free and divergence--free symmetric $(0,2)$-tensor on $(\RR^n , \delta_{\text{Euc}})$, one is required to find a positive function $\varphi$ tending to $1$ at infinity and a $1-$form $W$ satisfying
\begin{subequations}\label{CE}
\begin{align}
- {4 (n - 1) \over n - 2} \Delta \varphi + {n - 1 \over n} \rho^2 \varphi^{N-1} 
&= |\sigma + LW|^2 \varphi^{- N - 1}  &&[\text{{\footnotesize Lichnerowicz equation}}]
	\label{Liceq}
	\\
	\Div(L W) &= {n-1 \over n} \varphi^N d\rho,  &&[\text{{\footnotesize vector equations}}]
	\label{veceq}
\end{align}
\end{subequations}
where $N = {2n \over n - 2}$ and $L$ is the conformal Killing operator defined by
\begin{equation} \label{Killing}
(LW)_{ij} = \nabla_i W_j + \nabla_j W_i - {2 \delta_{ij} \over n} (\Div W).
\end{equation}
These equations are called the \textit{vacuum Einstein conformal constraint equations}, or simply the \textit{conformal equations}. Once such a solution $(\varphi ,W)$ exists, it follows that
\begin{equation} \label{parametre}
(g,k) :=  \bigg(\varphi^{N - 2} \delta_{\text{Euc}}, \, {\rho \over n} \varphi^{N - 2} \delta_{\text{Euc}} + \varphi^{-2} (\sigma + LW ) \bigg)
\end{equation}
is a solution to the vacuum constraint equations \eqref{constraint}. In this case, we remark that $\tr_g k = \rho$, therefore, $\rho$ is called the mean curvature in free data of this method. 
\subsection{Elliptic operators on weighted H\"{o}lder spaces} \label{subsection weighted Holder space} 
We next consider the elliptic operators used to study the conformal equations. For the proofs, we refer the reader to \cite{Diltsthesis}. 
\\
\\
Given an integer $l \ge 0$, a H\"{o}lder exponent $\alpha \in [ 0 , 1]$, and a decay exponent $ \beta> 0$, we will use the weighted H\"{o}lder spaces $C^{l , \alpha}_{-\beta}$ to capture asymptotic of functions and tensors near infinity. For $\alpha = 0$, we will write $C^{l}_{- \beta}$ instead of $C^{l , 0}_{- \beta}$. The weighted norm convention we are using is that the $C^{ s , \alpha}_{-\beta}$ norm is
given by
$$
\| f \|_{l , \alpha , -\beta} := \sum_{|s| \le l}\sup_{\RR^n} \big( \zeta^{|s| + \beta} |\del^s f|\big) + \sum_{|s| = l} \sup_{\RR^n} \Bigg(\zeta^{l + \beta + \alpha} \sup_{0 < |y - x| \le \zeta} \Big( {| \del^s f(y) - \del^s f(x)| \over |y - x|^\alpha}\Big) \Bigg)
$$
where in this context $\zeta$ is a positive function which equals $|x|$ outside the unit ball and $s$ is a multi-index. It will be clear from the context if the notation refers to a space of functions on $\RR^n$, or a space of sections of some bundle over $\RR^n$.
\begin{proposition}[Compact embedding for weighted H\"{o}lder spaces] \label{prop. properties of WH} If $l_1 + \alpha_1 > l_2 + \alpha_2$ and $\beta_1 > \beta_2 $ then the inclusion $C^{l_1 , \alpha_1}_{-\beta_1} \subset C^{l_2 , \alpha_2}_{-\beta_2}$ is compact.
\end{proposition}

\begin{proposition} [Weighted elliptic regularity for Laplacian] \label{prop. Laplacian} Let $V \ge 0$ be a function in $C^{l - 2 , \alpha}_{- 2 - \eps}$ with $l \ge 2, \, \alpha \in (0,1)$ and $  \eps> 0$. 
\begin{itemize}
	\item[(a)] $ \Delta - V : C^{l , \alpha}_{-\beta} \to C^{l - 2 , \alpha}_{- \beta- 2}$ is an isomorphism if and only if $ 0 < \beta< n - 2$. 
	\item[(b)] If $ u \in C^{0}_{\beta}$ and $\Delta u - V u \in C^{l,\alpha}_{- \beta - 2}$, then
	$$
	\| u \|_{l , \alpha , - \beta} \le c ( \| u \|_{C^{0}_{-\beta}} + \| \Delta u - V u \|_{l - 2 , \alpha , - \beta - 2} )
	$$
	for some constant $c > 0$ independent of $u$. 
\end{itemize} 
\end{proposition} 
Similarly, we also have the following proposition for the operator $\Div L$ appearing in the vector equations \eqref{veceq}, where $L$ is the conformal Killing  operator defined in \eqref{Killing}.
\begin{proposition}[Weighted elliptic regularity for vector Laplacian] \label{prop. vector Laplacian}
$ \Div L : C^{l , \alpha}_{-\beta} \to C^{l - 2 , \alpha}_{- \beta- 2} $ is an isomorphism if and only if $0 < \beta< n -2 $.
\end{proposition}
Finally, we give the theorem of existence and uniqueness of solutions to Lichnerowicz's equation on the Euclidean space, one of two main parts of the conformal equations: 
\begin{equation} \label{Licgeneral}
- {4 (n - 1) \over n - 2} \Delta u  + {n - 1 \over n} \rho^2 u^{N-1} 
= w^2 u^{- N - 1}.
\end{equation}

\begin{theorem}[Existence and uniqueness of solution to the Lichnerowicz equation] \label{theo. lichnerowicz}
If $\rho$ and $w$ are in  $C^{l - 2 , \alpha}_{- 1 -\beta/2}$ with $l \ge 2 , \, \alpha \in (0 , 1)$ and $ \beta\in (0 , n - 2)$, then the Lichnerowicz equation \eqref{Licgeneral} admits a unique positive solution $u$ satisfying $u - 1 \in C^{l , \alpha}_{- \beta}$.
\end{theorem}
\subsection{Spacetime Positive Energy Theorem}
Recall that the following definition and results are formulated for a general AF manifold. Here, however, we restrict our discussion to $\mathbb{R}^n$. Similarly, for simplicity and in line with our purposes, we consider the theorem only in the vacuum case.
\\
\\
Let $(\RR^n , g)$ be an AF manifold with 
\begin{equation} \label{mass well defined}
	g - \delta_{\text{Euc}} \in  C^{2, \alpha}_{ - {n - 2 \over 2} - \eps}
\end{equation}
for some $\eps >0$. The ADM mass of $(\RR^n , g)$ is defined by
$$
\madm (g) := {1 \over 2 (n - 2) \omega_{n - 1} } \lim_{r \to +\infty} \int_{|x| = r} \sum_{i,j = 1}^n (g_{ij,i} - g_{ii,j}) {x_j \over r} d \Mcal_0^{n - 1},
$$
where $\Mcal_0^{n - 1}$ is the $(n - 1)-$dimensional Euclidean Hausdorff measure and $\omega_{n-1}$ is the volume of the standard unit sphere in $\RR^n$. In particular, when $g = \varphi^{N - 2} \delta_{\text{Euc}}$ with $\varphi$ radial, the formula becomes
\begin{equation} \label{adm mass radial}
	\madm (g) = - {n - 1 \over 2(n - 2)} \lim_{r \to +\infty} (r^{n - 1} \varphi^\prime).
\end{equation}
Bartnik in \cite{Bartnik} showed that under the decay condition \eqref{mass well defined}, the mass is a geometric invariant. A long-standing conjecture in general relativity states that the ADM mass of a vacuum AF initial data set is positive unless it is a spacelike hypersurface of the Lorentz--Minkowski spacetime. This conjecture was proven to be true under suitable decay assumptions of $(g , k)$ at infinity and structural conditions on the manifold. The two most well-known results for this problem are the proof by Chruściel--Maerten \cite{ChruscielMaerten} for spin manifolds in arbitrary dimensions, and the proof by Eichmair \cite{Eichmair} for general manifolds in dimensions less than eight. We remark that the decay rate assumptions on $(g,k)$ at infinity in these results differ slightly. However, since our considered manifold $\mathbb{R}^n$ is spin for all $n \ge 3$, in the statement below we adopt the stronger of the two sets of decay rate assumptions.
\begin{pmt}[Chruściel--Maerten \cite{ChruscielMaerten} and Eichmair \cite{Eichmair}] \label{theorem SPMT}
Let $(\RR^n , g , k)$ be a vacuum AF initial data set. Assume that $(g - \delta_{\text{Euc}} , k) \in C^{2 , \alpha}_{-q+1} \times C^{1 , \alpha}_{-q}$ with $\alpha \in (0,1)$ and
\begin{equation} \label{decay-assumption-PET}
q >  \left\{\begin{array}{lr}
	{n \over 2} & \text{\quad if $3 \le n < 8$,}
	\\
	n - 2 & \text{\quad if $n \ge 8$~~~~~.}
\end{array}\right.
\end{equation}
Then the ADM  mass is non-negative. Moreover, if $\madm (g) = 0$, then $(\RR^n , g , k)$ is isometric to an AF spacelike hypersurface of $\LL^{n+1}$ with the second fundamental form $k$.
\end{pmt}
\section{Proof of Theorem \ref{main theorem}} \label{section proof of main theorem}
We are now ready to construct spacelike hypersurfaces of the Lorentz–Minkowski spacetime $\LL^{n+1}$ with a prescribed radial mean curvature. The construction will be carried out in three steps, corresponding to the subsections of this section. First, since the approach relies on the conformal method, we study how radial solutions of the conformal equations behave when the mean curvature 
$\rho$ is radial. Next, using this analysis, we give a detailed construction of an AF solution to the vacuum constraints that has the prescribed radial mean curvature $\rho$. Finally, we show that this solution has zero mass, and hence, by the rigidity part of the Spacetime PET, it is isometric to a spacelike hypersurface of $\LL^{n+1}$, which is our object. 
\subsection{Behavior of radial solutions to the conformal equation} \label{subsection behavior RCE}
Fixing a radial $\rho$, let us consider the conformal equations \eqref{CE} in the simple setting where $\sigma \equiv 0$. In this case, the system \eqref{CE} reduces to
\begin{subequations}\label{RCE}
	\begin{align}
		- \frac{4 (n - 1)}{n - 2} \Delta \varphi + \frac{n - 1}{n} \rho^2 \varphi^{N-1}
		&= |LW|^2 \varphi^{- N - 1}, \label{rLich}
		\\
		\Delta W_i + \frac{n - 2}{n} \partial_{i} \Big( \sum_{j =1}^n \partial_j W_j \Big) &= \frac{n-1}{n} \varphi^N \rho^\prime \frac{x_i}{r}. \label{rvector}
	\end{align}
\end{subequations}
Here and subsequently, $r$ is the usual Euclidean distance and we denote by $f^\prime$ the derivative of $f$ with respect to $r$. For our purpose, we restrict our attention to the set of all radial solutions to \eqref{RCE}. The following result plays a key role in our analysis.
\begin{proposition} \label{proposition construct solution}
Given $\alpha \in (0,1)$ and $\beta >0$, let $(\varphi, \rho)$ be radial functions such that $\varphi > 0$ and $(\varphi - 1 , \rho) \in C^{3 , \alpha}_{-  2 \beta + 2} \times C^{1, \alpha}_{ - \beta}$. Then $(\varphi , \rho)$ solves the conformal equations \eqref{RCE} if and only if  $\varphi^\p \ge 0$ and
\begin{equation} \label{mc identity}
	|\rho (r) | =  \left\{\begin{array}{lr}
		\sqrt{2 n N \varphi^{ - N + 1} \varphi^{\prime \prime}} & \text{ if $\varphi^\prime (r) = 0$,}
		\\
		{ (2n - 1) r^{-1} \varphi^{N/2} \varphi^\prime + N \varphi^{(N - 2)/2} (\varphi^\prime)^2 + \varphi^{N/2} \varphi^{\prime\prime}
			\over 
			\varphi^{N-1} \sqrt{(\varphi^\prime)^2 + (n-2) r^{-1} \varphi \varphi^\prime}} & \text{otherwise.}
	\end{array}\right.
\end{equation}   
\end{proposition}

\begin{proof} We will divide the proof into three steps.
\\
\\
\textbf{Step 1.} \textit{Solving the vector equations.} Letting
$$
f (r) := - \int_r^{+\infty} \varphi^N \rho^\prime \, d s,
$$
the vector equations \eqref{rvector} become
\begin{equation} \label{VE a}
	\Delta W_i + {n - 2 \over n} \del_{i} \Big( \sum_{j =1}^n \del_j W_j \Big) = { n - 1 \over n} \del_i f. 
\end{equation}
Differentiating \eqref{VE a} with respect to $i$ and summing all equations of the new system, we obtain
$
\Delta \Big( 2 \sum_{j = 1}^n \del_j W_j \Big) = \Delta f.
$
Therefore, by Proposition \ref{prop. Laplacian}(a), we have
$2\sum_{j = 1}^n \del_j W_j  =  f.$ Taking into account \eqref{VE a}, we get $\Delta W_i = {1 \over 2} \del_i f$. Since $f$ is radial, a direct computation shows that 
$$
W_i =  {x_i \over 2 r^n} \int_0^r s^{n-1} f \, d s.
$$
Thus, we have by definition
\begin{equation} \label{LW}
\aligned
(L W)_{ij} &=  - \Big( {\delta_{ij} \over n} - {x_i x_j \over r^2} \Big) \Big( f - {n \over r^n} \int_0^r s^{n-1} f \, ds \Big)
\\
& = - \Big( {\delta_{ij} \over n} - {x_i x_j \over r^2} \Big) \int_0^r s^n f^\prime \, ds
\\
& = - \Big( {\delta_{ij} \over n r^n} - {x_i x_j \over r^{n + 2}} \Big) \int_0^r s^n \varphi^N \rho^\prime \, ds
\endaligned
\end{equation} 
and so
\begin{equation} \label{module LW}
	|LW| = {{1 \over r^n} \sqrt{{n - 1 \over n}}} \Big| \int_0^r s^n \varphi^N \rho^\prime \, ds \Big|,
\end{equation}
which is a radial function.
\\
\\
\textbf{Step 2.} \textit{Solving the Lichnerowicz equation.}   It simplifies the argument, and causes no loss of generality, to assume $\varphi^\prime \ne 0$ almost everywhere. We first take \eqref{module LW} into the Lichnerowicz equation, it then follows that $(\varphi, \rho)$  solves the conformal equations \eqref{RCE} if and only if they satisfy 
\begin{equation} \label{one equation}
	- {4 n \over n- 2 } \Big(\varphi^{\prime \prime}  + {(n-1) \varphi^\prime \over r} \Big)+  \rho^2 \varphi^{N - 1} = {1 \over  r^{2 n}} \Big( \int_0^r s^n \varphi^N \rho^\prime \, ds \Big)^2 \varphi^{- N - 1}.
\end{equation}
Integrating by parts \eqref{one equation} we have
$$
- {4 n \over n- 2 } \Big(\varphi^{\prime \prime} + {(n-1) \varphi^\prime \over r} \Big) + \rho^2 \varphi^{N - 1}  =  {1 \over r^{2n}} \Big( r^n \varphi^N \rho - \int_0^r (s^n \varphi^N)^{\prime} \rho \, ds \Big)^2  \varphi^{-N - 1},
$$
equivalently,
\begin{equation} \label{derivative 0}
	2 r^n \varphi^N \rho  \int_0^r (s^n \varphi^N)^{\prime} \rho \, ds - \Big( \int_0^r (s^n \varphi^N)^{\prime} \rho \, ds \Big)^2 = 2N r^{2n} \varphi^{N + 1} \Big(\varphi^{\prime \prime} + {(n-1) \varphi^\prime \over r} \Big).
\end{equation} 
Next, multiplying \eqref{derivative 0} by $(r^n \varphi^N)^\prime / (r^n \varphi^N)^2$, we obtain
\begin{equation} \label{derivative condition}
	\Bigg( {1 \over r^n \varphi^N} \Big( \int_0^r (s^n \varphi^N)^{\prime} \rho \, ds \Big)^2 \Bigg)^\prime = 2N { (r^n \varphi^N)^\prime \over \varphi^{N - 1} } \Big(\varphi^{\prime \prime} + {(n-1) \varphi^\prime \over r} \Big).
\end{equation}
We observe here that the equations \eqref{derivative 0} and \eqref{derivative condition} are equivalent as long as $\varphi^\prime \ge 0$, which will be proven later, so we can continue our process without undue worry about equivalence among equations. Now, since
$
\lim\limits_{r\to 0} \Bigg( {1 \over r^n \varphi^N} \Big( \int_0^r (s^n \varphi^N)^{\prime} \rho \, ds \Big)^2 \Bigg) = 0,
$
the equation \eqref{derivative condition} is equivalent to
\begin{equation} \label{positive condition}
	\Big( \int_0^r (s^n \varphi^N)^{\prime} \rho \, ds \Big)^2 
	=
	2N \big( r^n \varphi^N \big) \Big(\int_0^r { (s^n \varphi^N)^\prime \over \varphi^{N - 1} } \Big(\varphi^{\prime \prime} + {(n-1) \varphi^\prime \over s} \Big) \, ds \Big).
\end{equation}
Therefore, assuming for the moment that 
\begin{equation} \label{condition 1}
	\int_0^r { (s^n \varphi^N)^\prime \over \varphi^{N - 1} } \Big(\varphi^{\prime \prime} + {(n-1) \varphi^\prime \over s} \Big) \, ds > 0 \quad \text{a.e in $\RR^n$,}
\end{equation}
we obtain
$$
\aligned
| \rho | &= \Bigg| { N \Bigg( \big( r^n \varphi^N \big) \Big( \int_0^r { (s^n \varphi^N)^\prime \over \varphi^{N - 1} } \Big(\varphi^{\prime \prime} + {(n-1) \varphi^\prime \over s} \Big) \, ds  \, ds\Big) \Bigg)^\prime
	\over 
	(r^n \varphi^N)^\prime 
	\sqrt{2N \big( r^n \varphi^N \big) \Big( \int_0^r { (s^n \varphi^N)^\prime \over \varphi^{N - 1} } \Big(\varphi^{\prime \prime} + {(n-1) \varphi^\prime \over s} \Big) \, ds\Big)} } \Bigg|
\\
&= {\Big| N (r \varphi) \big( r^{n - 1} \varphi^\prime \big)^\prime 
	+ N \int_0^r { (s^n \varphi^N)^\prime \over \varphi^{N - 1} } \Big(\varphi^{\prime \prime} + {(n-1) \varphi^\prime \over s} \Big) \, ds \Big|
	\over 
	\sqrt{\big( r^n \varphi^N \big) \Big( 2N \int_0^r { (s^n \varphi^N)^\prime \over \varphi^{N - 1} } \Big(\varphi^{\prime \prime} + {(n-1) \varphi^\prime \over s} \Big) \, ds \Big)} }.
\endaligned
$$
To simplify the formula, one can verify that
\begin{equation}\label{condition 2}
\bigg( {n \over 2} r^{n - 1} ( \varphi^2 )^\prime + {N \over 2} r^n (\varphi^\prime)^2 \bigg)^\p =  { (r^n \varphi^N)^\prime \over \varphi^{N - 1} } \Big(\varphi^{\prime \prime} + {(n-1) \varphi^\prime \over r} \Big),
\end{equation}
and so
$$
	\aligned
	|\rho| &=   {\big| 2 N (r \varphi) \big( r^{n - 1} \varphi^\prime \big)^\prime + n N r^{n - 1} ( \varphi^2 )^\prime + N^2 r^n (\varphi^\prime)^2 \big|
		\over 
		2 \sqrt{\big( r^n \varphi^N \big) \big( n N r^{n - 1} ( \varphi^2 )^\prime + N^2 r^n (\varphi^\prime)^2  \big)} }
	\\
	&= 	{\big| (2n - 1) r^{-1} \varphi^{N/2} \varphi^\prime + N \varphi^{(N - 2)/2} (\varphi^\prime)^2 + \varphi^{N/2} \varphi^{\prime\prime} \big|
		\over 
		\varphi^{N-1} \sqrt{(\varphi^\prime)^2 + (n-2) r^{-1} \varphi \varphi^\prime}}.
	\endaligned
$$
\textbf{Step 3.} \textit{$\varphi$ is increasing.} Recall that $\varphi^\p$ was assumed to be different from $0$ a.e in $\RR^n$ for simplicity. In view of \eqref{positive condition}--\eqref{condition 2}, we see that the necessary condition for $(\varphi , \rho)$ to be a solution to \eqref{RCE} is 
\begin{equation} \label{condition 3}
	\aligned
	n r^{n - 1} (\varphi^2)^\prime + N r^n (\varphi^\prime)^2  = r^{n-1} \varphi^\prime ( 2n \varphi + N r \varphi^\prime) > 0 \quad \text{a.e in $\RR^n$}.
	\endaligned
\end{equation}
We will show that this condition is equivalent to the fact that $\varphi$ is increasing. In fact, if $\varphi^\prime > 0$ a.e in $\RR^n$, \eqref{condition 3} is obvious. Conversely, assume that \eqref{condition 3} holds. Since $\varphi(0) > 0$, we have $ 2 n \varphi + N r \varphi^\prime > 0 $
near $0$. It then follows by \eqref{condition 3} that $\varphi^\prime \ge 0$ near $0$. Therefore, if $\varphi^\prime < 0$ somewhere, then there exists a convergent sequence $\{r_m\}$ such that $\varphi^\prime (r_m) < 0$ and $\varphi^\prime (r_m) \to 0$. This yields the contradiction $0 \le \varphi^\prime (r_m) \big( 2n \varphi (r_m) + N r_m \varphi^\prime (r_m) \big) < 0$. Therefore, we have $\varphi^\prime > 0$ a.e, and so the inequality \eqref{condition 3} holds as we assumed in \eqref{condition 1}. The proof is completed.
\end{proof}

\subsection{Existence of constraint solutions with prescribed radial mean curvature}
We now turn to the task of establishing the existence of vacuum constraint solutions whose mean curvature is given by a prescribed radial function $\rho$. Thanks to the conformal method and Proposition \ref{proposition construct solution}, we have the following result.
\begin{proposition}\label{proposition free mean curvature} Let $\rho$ be a radial function in $C^{1, \alpha}_{ - \beta} (\RR^n)$ with $\alpha \in (0 , 1)$ and $ \beta \in \big( 1 , {n \over 2} \big)$. Then the conformal equations \eqref{RCE} admit at least one radial positive function $\varphi$ satisfying $\varphi - 1 \in C^{3,\alpha}_{-2\beta +2}(\RR^{n})$. In particular, defining 
\begin{equation} \label{parameter new}
	\big( g_{ij}, k_{ij} \big) :=  \bigg( \varphi^{N - 2} \delta_{ij}, \, {\rho \over n} \varphi^{N-2} \delta_{ij} - \varphi^{-2} \Big( {\delta_{ij} \over n r^n} - {x_i x_j \over r^{n+2} }\Big) \int_0^r s^n \varphi^N \rho^\prime \, ds \bigg),
\end{equation} 
we obtain an AF solution $(g,k)$ to  the vacuum constraint \eqref{constraint} with $\tr_g k = \rho$.
\end{proposition}
\begin{proof}
We note that, by substituting the expression \eqref{LW} for $LW$ into \eqref{parametre}, it follows from the conformal method that the setting $(g,k)$ given in \eqref{parameter new} is automatically a solution to the constraint once $(\varphi, \rho)$ solves the conformal equations \eqref{RCE}. Therefore, it suffices to prove that \eqref{RCE} admit a radial positive solution $\varphi$ with $\varphi - 1 \in C^{3,\alpha}_{-2\beta +2}(\RR^{n})$. 
\\
\\
We first define the operator $T: [0 , 1] \times L^\infty \to L^\infty$ as follows. For any $\phi \in L^\infty$, by Proposition \ref{prop. vector Laplacian}, there exists a unique $W \in C^{2 , \alpha}_{ - \beta + 1}$ satisfying
\begin{equation} \label{vectoreq}
	\Div(L W) = {n-1 \over n} |\phi|^N d\rho,
\end{equation}
and hence, thanks to Theorem \ref{theo. lichnerowicz}, there exists a unique $\varphi>0$ such that  $\varphi - 1 \in C^{3, \alpha}_{- 2 \beta + 2}$ and
\begin{equation} \label{liceq}
	- {4 (n - 1) \over n - 2} \Delta \varphi  + {n - 1 \over n} t^{2N} \rho^2 \varphi^{N-1} 
	= |LW|^2 \varphi^{- N - 1}. 
\end{equation}
We define
\begin{equation} \label{def T}
T(t,\phi) := t \varphi.
\end{equation}
It is clear that a fixed point of $T(1, .)$ is a solution to the conformal equations \eqref{RCE}. In the spirit of the previous subsection, we will look for a fixed point of $T(1,.)$ in the subspace of radial functions
$$
RL^\infty := \{ f \in L^\infty ~|~ \text{$f$ is radial}\}.
$$
The following observations are the key to our arguments:
\\
\\
- Let $\Vcal : L^\infty \to C^{2 , \alpha}_{- \beta + 1}$ and $\Lcal : [0,1] \times C^{2 , \alpha}_{- \beta +1} \to C^{3 ,\alpha}_{-2 \beta + 2}$ be defined by
	$$
	\Vcal (\phi) := W, \qquad \Lcal (t , W)  : =\varphi - 1,
	$$
	where $W $ and $\varphi$ are determined by \eqref{vectoreq} and \eqref{liceq} respectively. Let $\Ical : C^{3 ,\alpha}_{-2 \beta + 2} \to L^\infty$ be the compact weighted H\"{o}lder embedding map given by Proposition \ref{prop. properties of WH}.  It is clear that 
	\begin{equation} \label{composition}
		T = t \big(\Ical \circ (1 + \Lcal) \circ \Vcal \big).
	\end{equation} 
	We have shown in the proof of Proposition \ref{proposition construct solution} that if $\phi$ is radial, then so is $|L \Vcal (\phi)|$. On the other hand, since the Laplace operator $\Delta$ is invariant under rotations, we deduce from the existence and uniqueness of solutions to the Lichnerowicz equation, guaranteed by Theorem \ref{theo. lichnerowicz}, that  if the source $(\rho , |LW|)$ is radial, then so is $\Lcal (t , W )$. Therefore, we can conclude by \eqref{composition} that $T(t,.)$ maps the subspace $RL^\infty$ into itself.
\\
\\
- If $T(t , \phi) = \phi$ with $(t,\phi) \in (0,1] \times RL^\infty$, then $\phi / t$ is a radial solution to the conformal equations \eqref{RCE} associated with the seed data $(\delta_{\text{Euc}} , t^N \rho)$. By Proposition \ref{proposition construct solution}, it follows that $\phi/t$ is increasing, and hence $\|\phi / t \|_{L^\infty} = 1$. In particular, the set $
K = \big\{ (t ,\phi) \in (0,1] \times RL^\infty ~\big|~  T(t , \phi) = \phi \big\}$ is bounded.
\\
\\
From these observations, once $T$ is proven to be  continuous and compact in $[0,1] \times RL^\infty$, the Leray--Schauder fixed point ensures that $T(1 , .)$ has a fixed point in $RL^\infty$ which is exactly what we desire. Moreover, in view of \eqref{composition}, since $\Vcal$ is continuous and  $\Ical$ is continuous and compact, it follows that $T$ will be continuous and compact once we establish the continuity of $\Lcal$. Therefore, it remains to prove that $\Lcal$ is a continuous operator. The argument we give here is essentially the same as in \cite{Maxwellcompact, NguyenFPT}, which give the corresponding  result for compact manifolds. 
\\
\\
In fact, we define $F ( t, W , \psi) : [0,1] \times C^{2 , \alpha}_{ - \beta + 1} \times C^{3 , \alpha}_{- 2 \beta + 2} \to C^{1 , \alpha}_{- 2\beta}$ by 
$$
F ( t , W , \psi) := - {4 (n - 1) \over n - 2} \Delta (\psi + 1)  + {n - 1 \over n} t^{2N} \rho^2 (\psi + 1)^{N-1} 
- |LW|^2 (\psi + 1)^{- N - 1}
$$
It is clear that $F$ is $C^1$ map and  $F(t , W , \Lcal(t,W)) = 0$ for all $(t , W) \in [0,1] \times C^{2 , \alpha}_{ - \beta + 1}$. A standard computation shows that the Fr\'echet derivative of $F$ with respect to $\psi$ is given by
$$
F_\psi (t , W ) (u) = - {4 (n - 1) \over n - 2} \Delta u  + {(n - 1) (N - 1) \over n} t^{2N} \rho^2 (\psi + 1)^{N-2} u
+ (N + 1) |LW|^2 (\psi + 1)^{- N - 2} u
$$
It follows that $F_\psi \in C\big( [0,1] \times C^{2 , \alpha}_{ - \beta + 1} , L (C^{3 , \alpha}_{- 2 \beta + 2}  , C^{1 , \alpha}_{- 2 \beta}) \big)$, where we denote $L (C^{3 , \alpha}_{- 2 \beta + 2}  , C^{1 , \alpha}_{- 2 \beta})$
the Banach space of all linear continuous maps from $ C^{3 , \alpha}_{- 2 \beta + 2} $ into $C^{1 , \alpha}_{- 2 \beta}$. In particular, setting $\psi_0 = \Lcal(t , W)$ we have
$$
F_{\psi_0} (t , W ) (u) = - {4 (n - 1) \over n - 2} \Delta u  + \Bigg( {(n - 1) (N - 1) \over n} t^{2N} \rho^2 (\psi_0 + 1)^{N-2} 
+ (N + 1) |LW|^2 (\psi_0 + 1)^{- N - 2} \Bigg) u
$$
Since 
$$
{(n - 1) (N - 1) \over n} t^{2N} \rho^2 (\psi_0 + 1)^{N-2} 
+ (N + 1) |LW|^2 (\psi_0 + 1)^{- N - 2} \ge 0,
$$
it follows by Proposition \ref{prop. Laplacian}(a) that
$F_{\psi_0} (t , W ) : C^{3 , \alpha}_{- 2 \beta + 2} \to C^{1 , \alpha}_{- 2 \beta}$ is an isomorphism. Therefore, the implicit function theorem implies that $\Lcal$ is a $C^1$-function in a neighborhood of $(t , W)$, which deduces the continuity of $\Lcal$. The proof is completed.
\end{proof}
\subsection{Vanishing ADM mass and spacelike hypersurfaces in $\LL^{n+1}$}
We are now at the final stage, namely to show that the vacuum constraint solutions in Proposition \ref{proposition free mean curvature} have zero ADM mass, and so, they correspond to AF spacelike hypersurfaces of $\LL^{n+1}$. The following result completes the proof of Theorem \ref{main theorem}.
\begin{proposition} \label{proposition null mass}
Let $(\RR^n, g,k)$ be an AF vacuum constraint solution given by \eqref{parameter new} in Proposition \ref{proposition free mean curvature}. Assume that $| \rho | \sim c r^{- q}$ at infinity for some constant $c>0$ and decay exponent $q \in \big({n+2 \over 4}, n\big)$. Then $(g - \delta_{\text{Euc}} , k) \in C^{3, \alpha}_{- 2 q + 2} \times C^{1, \alpha}_{ - q}$ and moreover
\begin{itemize}
	\item[(i)] if $q < {n \over 2}$, then $\madm (g) = - \infty$,
	\item[(ii)] if $q = {n \over 2}$ , then $- \infty < \madm (g) < 0$,
	\item[(iii)] if $q > {n \over 2}$, then $\madm (g) = 0$, and hence, $(\RR^n , g , k)$ is isometric to an AF spacelike hypersurface of $\LL^{n+1}$ as long as $(q,n)$ fulfills \eqref{decay-assumption-PET}. In particular, Theorem \ref{main theorem} holds.
\end{itemize}
\end{proposition} 
\begin{proof}
We recall that the radial functions $\rho$ and $\varphi$ in the definition \eqref{parameter new} of $(g,k)$ solves the conformal equations \eqref{RCE}. On one hand, by Proposition \ref{proposition construct solution}, it follows that $\varphi^\p \ge 0$ and $(\rho , \varphi)$ must satisfy Identity \eqref{mc identity}. Then, provided that $|\rho| \sim c r^{-q}$, this identity gives us
$$
	\lim_{r \to +\infty} \bigg( { \big| (2n - 1) r^{-1} \varphi^{N/2} \varphi^\prime + N \varphi^{(N - 2)/2} (\varphi^\prime)^2 + \varphi^{N/2} \varphi^{\prime\prime} \big|
		\over 
		r^{-q} \varphi^{N-1} \sqrt{(\varphi^\prime)^2 + (n-2) r^{-1} \varphi \varphi^\prime}} \bigg)= c.
$$
On the other hand, since $\varphi - 1 \in C^{3 , \alpha}_{-  2 \beta + 2}$, it is not difficult to check that
$$
\aligned
\lim_{r \to +\infty} \bigg( { \big| (2n - 1) r^{-1} \varphi^{N/2} \varphi^\prime + N \varphi^{(N - 2)/2} (\varphi^\prime)^2 + \varphi^{N/2} \varphi^{\prime\prime} \big|
	\over 
	r^{-q} \varphi^{N-1} \sqrt{(\varphi^\prime)^2 + (n-2) r^{-1} \varphi \varphi^\prime}} \bigg)
&= {1 \over \sqrt{n - 2}}\lim_{r \to +\infty} { \Big| (r^{2n - 1} \varphi^\prime )^\prime \Big| 
	\over 
	r^{n - 1 - q}\sqrt{ r^{2n - 1} \varphi^\prime}}
\\
&= {2 (n - q) \over \sqrt{n - 2}}\lim_{r \to +\infty} \Bigg| { \big( \sqrt{ r^{2n - 1} \varphi^\prime} \big)^\prime \over ( r^{n - q} )^\prime }\Bigg|.
\endaligned
$$
Combining these two facts, we get 
$$
\lim\limits_{r \to +\infty} \Big| { \big( \sqrt{ r^{2n - 1} \varphi^\prime} \big)^\prime \over ( r^{n - q} )^\prime }\Big|= {c \sqrt{ n - 2} \over 2(n - q)}.
$$
Applying L'H\^{o}pital's rule to the ratio above, it follows that
\begin{equation} \label{calculate mass}
\lim_{r \to +\infty} \Big( {\varphi^\prime \over r^{- 2q + 1}} \Big) = \Bigg(\lim_{r \to +\infty} {\sqrt{r^{2n - 1}\varphi^\prime } \over r^{n - q}}\Bigg)^2
= \Bigg( \lim_{r \to +\infty} \Bigg| { \big( \sqrt{r^{2n - 1}\varphi^\prime} \big)^\prime \over ( r^{n - q} )^\prime }\Bigg| \Bigg)^2
= {c^2 (n - 2) \over 4(n - q)^2}.
\end{equation}
In particular, this tells us that $
\varphi - 1 = - \int_r^{+\infty} \varphi^\prime \, ds \in C^0_{-2 q + 2}$,
and hence, thanks to the Lichnerowicz equation and Proposition \ref{prop. Laplacian}(b), we deduce that $\varphi \in C^{3,\alpha}_{-2q+2}$, which implies $(g - \delta_{\text{Euc}} , k) \in C^{3 , \alpha}_{- 2 q + 2} \times C^{1, \alpha}_{ - q}$ by definition.
\\
\\
Now, taking \eqref{calculate mass} into the formula \eqref{adm mass radial} yields
$$
\madm (g) = - {c^2 (n - 1) \over 8 (n - q)^2} \lim\limits_{r \to +\infty} r^{n - 2q  }.
$$
From this fact, together with the Spacetime PET (which is needed only for case (iii)), we conclude cases (i--iii). The proof is completed.
\end{proof}
We have established the existence of classical solutions to the equation \eqref{BI} when $\rho$ is radial by using the conformal method and the Spacetime PET. To conclude the article and provide further discussion of the approach, we make the following remarks.
\\
\\
Regarding the sign of mass, we emphasize that the appearance of negative mass in Proposition \ref{proposition null mass} does not contradict the Spacetime PET. This is because the decay rate of $k$ at infinity in the first two cases of the proposition is either critical or subcritical relative to the decay assumptions imposed on the second fundamental form in the theorem. Combined with the construction of Chruściel \cite{Chruscielletter}, which provides an example of vacuum constraint solutions with negative mass in the case where $g - \delta_{\text{Euc}}$ decays subcritically at infinity, these results demonstrate that the decay conditions on both $g$ and $k$ stated in the theorem are sharp.
\\
\\
It is also worth noting that Proposition \ref{proposition null mass} can be strengthened to state that, in all three cases, $(\RR^n,g,k)$ is isometric to a spacelike hypersurface of Lorentz–Minkowski spacetime, independently of the value of the mass. Consequently, Theorem \ref{main theorem} remains valid under the weaker assumption $q>2$ in all dimensions. In fact, since $\rho$ and $\varphi$ are radial, $(g,k)$ is, by definition, regular and spherically symmetric. Then, it follows by Birkhoff’s theorem that $(\RR^n,g,k)$ is automatically an AF initial data set in 
$\LL^{n+1}$, without any need to consider its mass. Nevertheless, as noted in the Introduction, we prefer the current statement, since it seems to apply to a broader class of $\rho$ beyond the radial case.
\\
\\
Finally, we remark that solutions to the electrostatic Born--Infeld \eqref{BI} in our result are classical and spacelike, representing an improvement in solvability compared to those in \cite{BonheureAveniaPomponio BIEquationRadialCharge, BonheureIacopetti BIEquationSmallCharge}. Conversely, the regularity assumption on $\rho$ in our result is a bit stronger than in existing works, where $\rho$ is only required to be Sobolev regular. More precisely, Sobolev regularity does not present any difficulty at the stage of applying the conformal method to construct initial data sets with prescribed mean curvature, but  H\"{o}lder  regularity for $\rho$ is needed in order to satisfy the assumptions in the spacetime PET version of Chruściel--Maerten \cite{ChruscielMaerten} and Eichmair \cite{Eichmair}. Therefore, if the regularity assumptions in the spacetime PET could be weakened, the assumptions in Theorem \ref{main theorem} would immediately become less restrictive as well. This appears to be a feasible direction, inspired by the work of Lee and LeFloch \cite{LeeLeFloch}, who showed that the Positive Mass Theorem continues to hold for metrics with only Sobolev regularity. If a similar result can be extended to the spacetime PET, then Theorem \ref{main theorem} could be further strengthened by assuming only Sobolev regularity on $\rho$.

\end{document}